\documentclass[a4paper,UKenglish,cleveref, autoref]{lipics-v2019}

\usepackage[utf8]{inputenc}
\usepackage{amsmath}
\usepackage{amsfonts}
\usepackage{amsthm}
\usepackage{amssymb}

\usepackage{tikz}
\usepackage{tkz-tab}
\usepackage{caption}
\usepackage{latexsym}
\usepackage{subcaption}
\usetikzlibrary{positioning,shapes,fit,arrows,calc}

\usepackage{algorithm}
\usepackage[noend]{algpseudocode}

\usepackage{stmaryrd}
\usepackage{enumitem}
\usepackage{tabularx}
\usepackage{hyperref}
\usepackage{todonotes}

\newcommand{\suparrow}{{\uparrow}}

\newcommand{\inc}{{+}{+}}
\newcommand{\dec}{{-}{-}}

\newcommand{\ceil}[1]{\left\lceil #1 \right\rceil}

\newlist{casesp}{enumerate}{3} 
\setlist[casesp]{align=left, 
                 listparindent=\parindent, 
                 parsep=\parskip, 
                 font=\normalfont\bfseries, 
                 leftmargin=0pt, 
                 labelwidth=0pt, 
                 itemindent=.4em,labelsep=.4em, 
                 partopsep=0pt, 
                 }
\setlist[casesp,1]{label=Case~\arabic*:,ref=\arabic*}
\setlist[casesp,2]{label=Case~\thecasespi.\roman*:,ref=\thecasespi.\roman*}
\setlist[casesp,3]{label=Case~\thecasespii.\alph*:,ref=\thecasespii.\alph*}

\title{The Well Structured Problem for Presburger Counter Machines} 

\author{Alain Finkel}{LSV, ENS Paris-Saclay, CNRS, Universit\'{e}
Paris-Saclay, France \and UMI ReLaX}{alain.finkel@ens-paris-saclay.fr}{}{}
\author{Ekanshdeep Gupta}{Chennai Mathematical Institute, Chennai, India \and UMI ReLaX} {ekanshdeep@cmi.ac.in}{}{}{}

\authorrunning{A. Finkel and E. Gupta}

\Copyright{Alain Finkel and Ekanshdeep Gupta}



\category{}

\relatedversion{}

\supplement{}

\acknowledgements{The work reported was carried out in the framework of ReLaX, UMI2000 (ENS
Paris-Saclay, CNRS, Univ. Bordeaux, CMI, IMSc). This work was also supported
by the grant ANR-17-CE40-0028 of the French National Research Agency ANR
(project BRAVAS).}

\nolinenumbers 

\hideLIPIcs  

\EventEditors{Arkadev Chattopadhyay and Paul Gastin}
\EventNoEds{2}
\EventLongTitle{39th IARCS Annual Conference on Foundations of Software Technology and Theoretical Computer Science (FSTTCS 2019)}
\EventShortTitle{FSTTCS 2019}
\EventAcronym{FSTTCS}
\EventYear{2019}
\EventDate{December 11--13, 2019}
\EventLocation{Bombay, India}
\EventLogo{}
\SeriesVolume{150}
\ArticleNo{40}

\begin{document}

\maketitle

\begin{abstract}
We introduce the \emph{well structured problem} as the question of whether a model (here a counter machine) is well structured (here for the usual ordering on integers). We show that it is undecidable for most of the (Presburger-defined) counter machines except for Affine VASS of dimension one. However, the \emph{strong well structured problem} is decidable for \emph{all} Presburger counter machines. While Affine VASS of dimension one are not, in general, well structured, we give an algorithm that computes the set of predecessors of a configuration; as a consequence this allows to decide the well structured problem for $1$-Affine VASS.
\end{abstract}

\section{Introduction}
{\bf Context: }
Well Structured Transition Systems (WSTS) \cite{F90,FINKEL200163} are a well-known model to solve termination, boundedness, control-state reachability and coverability problems. 
It is well known that Petri nets and Vector Addition Systems with States (VASS) are WSTS and that Minsky machines are not WSTS. But the characterization of counter machines which are well structured (resp. with strong monotony) is surprisingly unknown. Moreover, given a counter machine, can we decide whether it is well structured (resp. with strong monotony)? These questions are relevant since a positive answer could allow to verify \emph{particular instances} of undecidable models like Minsky machines and counter machines.
In this paper, we consider \emph{Presburger counter machines (PCM)} where each transition between two control-states is labelled by a Presburger formula which describes how each counter is modified by the firing of the transition. The PCM model includes Petri nets, Minsky machines and most of the counter machine models studied in the literature, for example counter machines where transitions between control-states are given by affine functions having Presburger domains \cite{DBLP:conf/cav/BoigelotW94,DBLP:conf/fsttcs/FinkelL02}. \\

\noindent {\bf Affine VASS: }
In an Affine VASS (AVASS), transitions between control-states are labelled by affine functions whose matrices have elements in $\mathbb{Z}$ (and not in $\mathbb{N}$ as usual). AVASS extends VASS (where transitions are translations) and positive affine VASS (introduced as self-modified nets in \cite{DBLP:conf/icalp/Valk78} and studied as affine well structured nets in \cite{FMP-wstsPN-icomp}. \cite{DBLP:conf/fsttcs/BonnetFP12} extends the Rackoff technique to AVASS where all matrices are larger than the identity matrix: for this subclass, coverability and boundedness are shown in EXPSPACE. The variation of VASS which may go below $0$, called $\mathbb{Z}$-VASS, is studied in \cite{DBLP:conf/rp/HaaseH14} and for their extension, \emph{$\mathbb{Z}$-Affine} VASS, reachability is shown NP-complete for VASS with resets, PSPACE-complete for VASS with transfers and undecidable in general \cite{BR19,DBLP:conf/concur/BlondinHM18}; let us remark that all $\mathbb{Z}$-Affine VASS have \emph{positive} matrices.


Moreover AVASS allow the simulation of the zero-test so they are at least as expressive as Minsky machines. But for dimension one, AVASS are more expressive than Minsky machines: in fact, $Post^*$ is computable as a Presburger formula for $1$-counter Minsky machines but this is not the case for $1$-AVASS which can generate the set of all the powers of $2$ (this set is not the solution of any Presburger formula).

The computation of the set $Pre^*$ of all predecessors of a configuration is effective for $2$-VASS (extended with one zero-test and resets) \cite{FLS-fsttcs18} as a Presburger formula and for pushdown automata \cite{BEFMRWW-ipl2000} as a regular language. But the computation of $Pre^*$ fails for $3$-VASS and for Pushdown VAS since $Pre^*$ is neither semilinear nor regular \cite{DBLP:conf/icalp/LerouxST15}. \\

\noindent
{\bf Our contributions:}

We introduce two new problems related to well structured systems and Presburger counter machines. The so-called \emph{well structured problem}: (1) given a PCM, is it a WSTS (for the usual ordering on integers) ? and the \emph{strong well structured problem}: (2) given a PCM, is it a WSTS with \emph{strong} monotony? 

We prove that  the well structured problem is undecidable for PCM even if restricted to dimension one ($1$-PCM) with just Presburger functions (i.e. piecewise affine functions); undecidability is also verified for Affine VASS in dimension two ($2$-Affine VASS). The undecidability proofs use the fact that Minsky machines can be simulated by both $1$-PCM and $2$-Affine VASS. However, we prove the decidability of the well structured problem for $1$-Affine VASS (which subsumes $1$-Minsky machines). 

Since the strong monotony can be expressed as a Presburger formula, the strong well structured problem (with the usual ordering on integers) is decidable for \emph{all} PCMs; moreover, we show that the decidability of the strong well structured problem can be extended to PCMs equipped with \emph{any} quasi ordering defined by a Presburger formula. 

Most of these results are summarised below: \\

\begin{tabularx}{\textwidth}{|l|X|X|}
\hline
    & Well Structured Problem & Strong Well Structured Problem \\ \hline
PCM & U  & {\bf D} \\ \hline
Functional $1$-PCM & {\bf U} [Theorem \ref{1-dim-pcs-wsts}] & D \\ \hline
$2$-AVASS &  U & D  \\ \hline
$2$-Minsky machines & {\bf U} [Theorem  \ref {2-ctr-minsky-machine}] & D  \\ \hline
$1$-AVASS & {\bf D} [Theorem \ref{1-aff-WSTS-dec}] & D  \\ \hline
\end{tabularx}\\

We give an algorithm that computes $Pre^*$ of a $1$-AVASS and this extends a similar known result for $1$-Minsky machines and $1$-VASS (and for pushdown automata \cite{BEFMRWW-ipl2000}). The computation of $Pre^*$ allows us to give a simple proof that reachability and coverability are decidable for $1$-AVASS (in fact reachability is known to be PSPACE-complete for polynomial one-register machines  \cite{FGH-mfcs13} which contains $1$-AVASS). Moreover, the computation of $Pre^*$ allows to decide the well structured problem for $1$-AVASS. These results are summarised below: \\

\begin{tabularx}{\textwidth}{|X|X|X|}
\hline
 & Reachability & Coverability \\ \hline
$1$-PCM (functional ) & U & {\bf U} [Corollary \ref{1-pcm-cov-undec}]  \\ \hline
$1$-AVASS & {\bf D} [Corollary \ref{1-avass-reachability}] & D \\ \hline
$d$-totally positive AVASS & {\bf D} [Theorem \ref{tot-pos-avass-reach-dec}] & D \\ \hline
$d$-positive AVASS ($d \geq 2$) & {\bf U} [Theorem \ref{pos-2-avass-reach-undec}] & {\bf D} [WSTS] \\ \hline
$2$-AVASS & U & {\bf U} [Corollary \ref{2-avass-cov-undec}] \\ \hline
\end{tabularx} \\

\noindent{\bf Outline:} We introduce in Section $2$ two models, well structured transition systems (WSTS) and Presburger counter machines (PCM); we show that the property for an ordering to be well is undecidable. Section $3$ analyses the decidability of the well structured problems for many classes of PCM and Affine VASS. Section $4$ studies the decidability of reachability and coverability for the classes studied in Section $3$.


\section{Counter machines and WSTS} \label{section:counter-machines}

A relation $\leq$ on a set $E$ is a \emph{quasi ordering} if it is reflexive and transitive; it is an ordering if moreover $\leq$ is antisymetric.
A quasi ordering $\leq$ on $E$ is a \textit{well quasi ordering (wqo)} if for all infinite sequences of elements of $E$, $(e_i)_{i \in \mathbb{N}}$, there exists two indices $i < j$ such that $e_i \leq e_j$.
For an ordered set $(E, \leq)$ and a subset $X \subseteq E$, the upward closure of $X$ denoted by $\suparrow X$ is defined as follows: $\suparrow X = \{ x \mid \exists y \in X \text{ such that } y \leq x\}$. 
$X$ is said to be \textit{upward closed} if $X = \suparrow X$.

\subsection{Arithmetic counter machines}
A \emph{$d$-dim arithmetic counter machine (short,  d-arithmetic counter machine or an arithmetic counter machine)} is a tuple $M=(Q, \Phi, \rightarrow)$ where $Q$ is a finite set of control-states, $\Phi$ is a set of logical formulae with $2d$ free variables $x_1, ..., x_d, x'_1, ..., x'_d$ and $\rightarrow \subseteq Q \times \Phi \times Q$ is the transition relation between control-states. We can also without loss of generality assume that $\rightarrow$ covers $\Phi$, i.e. $\Phi$ does not have unnecessary formulae. A configuration of $M$ refers to an element of $Q \times \mathbb{N}^d$. The operational semantics of a $d$-arithmetic counter machine $M$  is a transition system $S_M=(Q \times \mathbb{N}^d,\rightarrow)$ where $\rightarrow \subseteq (Q \times \mathbb{N}^d) \times (Q \times \mathbb{N}^d) $ is the transition relation between configurations. For a transition $(q,\phi,q')$ in $M$, we have a transition $(q; x_1, ..., x_d) \rightarrow (q'; x'_1, ..., x'_d)$ in $S_M$ iff $ \phi(x_1, ..., x_d, x'_1, ..., x'_d)$ holds. Note that we are slightly abusing notation by using the same $\rightarrow$ for both $M$ and $S_M$.
We may omit $\Phi$ from the definition of a counter machine if it is clear from context. 

A $d$-dim arithmetic counter machine $M$ with \emph{initial configuration} $c_0$ is defined by the tuple $M=(Q, \Phi, \rightarrow, c_0)$ where $(Q, \Phi, \rightarrow)$ is a $d$-arithmetic counter machine and  $c_0 \in Q \times \mathbb{N}^d$ is the initial configuration.
An arithmetic counter machine is \emph{effective} if  the transition relation is decidable (there is a decidable procedure to determine if there is a transition $x \rightarrow y$ between any two configurations $x, y$) and this is the case when it is given by an algorithm, a recursive relation, or decidable first order formulae (for instance Presburger formulae). An arithmetic counter machine is said to be \emph{functional} if each formula in $\Phi$ that labels a transition in $M$ defines a partial function.

Most usual counter machines can be expressed with Presburger formulae. It is well known that Presburger arithmetic with congruence relations without quantifiers is equivalent in expressive power to standard Presburger arithmetic \cite{Haase:2018:SGP:3242953.3242964}.
\begin{definition}
A \emph{Presburger counter machine (PCM)} is an arithmetic counter machine $M=(Q,\Phi, \rightarrow)$ such that $\Phi$ is a set of Presburger formulae with congruence relations without quantifiers.
\end{definition}

\begin{proposition} \label{pcm-functional-prop}\cite{DBLP:journals/jancl/DemriFGD10}
The property for a $d$-dim PCM to be \emph{functional} is decidable in NP.
\end{proposition}

\begin{proof}
	Let $M=(Q,\Phi,\rightarrow)$ be a given $d$-dim PCM.
	Functionality can be expressed in Presburger arithmetic as follows: 
	\begin{align*}
	\bigwedge\limits_{\phi \in \Phi} &( \forall x_1 ... \forall x_d \forall x'_1 ... \forall x'_d \forall x''_1 ... \forall x''_d \\
	&(\phi(x_1, ..., x_d, x'_1, ..., x'_d) \land \phi(x_1, ..., x_d, x''_1, ..., x''_d) \implies \bigwedge_{i=1}^d x'_i = x''_i)
	\end{align*}
	
	Hence, the validity of this formula can be decided, and so, functionality is decidable.
\end{proof}

\emph{Minsky machines} with $d$ counters are $d$-PCM $M=(Q,\Phi, \rightarrow)$ where $\Phi$ consists of either translations with upwards closed guards, or formulae of the form $\land_{i=1}^d (x_i = x'_i) \land x_k=0$ for varying $k$ (zero-tests).
\emph{Vector Addition Systems with States (VASS)} are Minsky machines without zero-tests.
An \emph{Affine VASS} with $d$ counters ($d$-AVASS) is a $d$-PCM where each transition is labelled by a formula equivalent to an affine function of the form $f(x) = Ax + b$ where $A \in M_d(\mathbb{Z})$ is a $d \times d$ matrix over $\mathbb{Z}$ and $b \in \mathbb{Z}^d$. The domain of such a function would be the (Presburger) set of all $x \in \mathbb{N}^d$ such that $Ax + b \in \mathbb{N}^d$. For convenience, we will denote $d$-AVASS transitions by a pair $(A,b) \in M_d(\mathbb{Z}) \times \mathbb{Z}^d$.
Note that AVASS is an extension of VASS where transitions are not labelled by vectors but by affine functions $(A_i,b_i)$.
Let us define positive and totally-positive AVASS. A \emph{positive} AVASS $S$ is an AVASS such that every matrix $A_i$ of $S$ is positive. This model has been studied for instance in \cite{FMP-wstsPN-icomp}.
A \emph{totally-positive} AVASS $S$ is a positive AVASS such that every vector $b_i$ of $S$ is positive. For totally positive AVASS, an instance of the boundedness problem has been shown decidable in  \cite{FMP-wstsPN-icomp}. Note that we say something is positive if it is greater than or equal to $0$, not strictly greater than $0$.

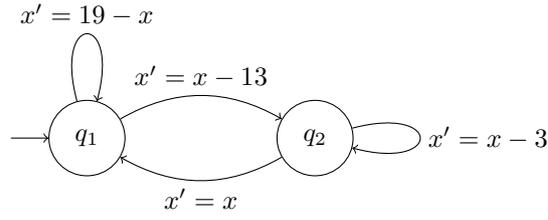
\begin{figure}
    \centering
        \resizebox{!}{!}{
\begin{tikzpicture}
\tikzstyle{every node}+=[inner sep=0pt]
\node(s1) at (0,0) [circle, draw, minimum size=1cm] {$q_1$};
\node(s2) at (3,0) [circle, draw, minimum size=1cm] {$q_2$};

\draw [->] (s1) to [out=30, in=150] node [midway, above = 0.1] {$x' = x - 13$} (s2);

\draw [->] (s1) to [out=105, in=75, looseness=12] node [midway, above = 0.1] {$x' = 19 - x$} (s1);

\draw [->] (s2) to [out=15, in=345, looseness=12] node [midway, right = 0.1] {$x' = x - 3$} (s2);

\draw [->] (s2) to [out=210, in=330] node [midway, below = 0.1] {$x' = x$} (s1);

\draw [->] (-1,0) -- (s1);
\end{tikzpicture}
        }
\caption{The counter machine $M_1$ }
\label{figure:example-avass}
\end{figure}

\begin{example} \label{example:avass}
The machine $M_1$ in Figure \ref{figure:example-avass} is a $1$-AVASS but it is not a $1$-VASS because there is a negative transition from $q_1$ to $q_1$.
\end{example}

\begin{proposition} \cite{DBLP:journals/jancl/DemriFGD10}
Checking whether a given PCM is a VASS, AVASS, positive AVASS or a totally positive AVASS is decidable.
\end{proposition}

\subsection{Well structured transition systems} \label{subsection:wsts}

A \emph{transition system} is a tuple $S = (X, \rightarrow)$ where $X$ is a (potentially infinite) set of configurations and $\rightarrow \subseteq X \times X$ is the transition relation between configurations. We denote by $ \xrightarrow{*} $ the reflexive and transitive closure of $\xrightarrow{}$. For a subset $S \subseteq X$, we denote by $Pre(S) := \{t \mid t \rightarrow s \text{ for some } s \in S \}$, and $Pre^*(S) := \{ t \mid t \xrightarrow{*} s \text{ for some } s \in S \}$. Similarly for $Post(S)$ and $Post^*(S)$.

An \emph{ordered transition system} $S = (X, \rightarrow,\leq)$ is a transition system $ (X, \rightarrow)$ with a quasi-ordering $\leq$ on $X$. Given two configurations $x,y \in X$, $x$ is said to \emph{cover} $y$ if there exists a configuration $y' \geq y$ such that $x \xrightarrow{*} y'$. An ordered transition system $S=(X, \rightarrow, \leq)$ is \emph{monotone}, if for all configurations $s,t,s' \in X$ such that $s \rightarrow t$, $s' \geq s$ implies that $s'$ covers $t$.
$S$ is \emph{strongly monotone} if for all configurations $s,t,s' \in X$ such that $s \rightarrow t$, $s' \geq s$ implies that there exists $t' \geq t$ such that $s' \rightarrow t'$.

\begin{definition} \cite{FINKEL200163}
A \emph{well structured transition system} (WSTS) is an ordered transition system $S = (X, \rightarrow, \leq)$ such that $(X,\leq)$ is a wqo and $S$ is monotone.
\end{definition}

The \textit{coverability problem} is to determine, given two configurations $s$ and $t$, whether there exists a configuration $t'$ such that $s \xrightarrow{*} t' \geq t$ ($s$ covers $t$). This problem is one often studied alongside well-structuredness.

Let us consider the usual wqo $\leq$ on $Q \times \mathbb{N}^d$ associated with a $d$-counter machine $M=(Q,\rightarrow)$:
$(q_1; x_1, x_2, ..., x_d) \leq (q_2; y_1, ..., y_d) \iff (q_1 = q_2) \land (\land_{i=1}^d x_i \leq y_i)$.

We say that an arithmetic counter machine $M=(Q,\Phi, \rightarrow)$ is \emph{well structured} (or is a WSTS) iff its associated transition system $S_M$ is a WSTS under the usual ordering.
Since the usual ordering on $(Q \times \mathbb{N}^d,\leq)$ is a wqo, let us remark that the associated ordered transition system $S_M = (Q \times \mathbb{N}^d, \rightarrow, \leq)$ is a WSTS iff $S_M$ is monotone. 

Given a counter machine $M = (Q, \rightarrow)$, the \textit{control-state reachability problem} is that given a configuration $(q; n_1, ..., n_d)$, and a control-state $q'$ whether there exist values of counters $(m_1, ..., m_d)$ such that $(q; n_1, ..., n_d) \xrightarrow{*} (q'; m_1, ..., m_d)$. In this case, we often say that $q'$ is \emph{reachable} from $(q; n_1, ..., n_d)$.

We introduce two new problems related to WSTS and Presburger counter machines. 
\begin{itemize}
\item The \emph{well structured problem}: given a PCM, is it a WSTS? 
\item The \emph{strong well structured problem}: given a PCM, is it a WSTS with strong monotony? 
\end{itemize}

\begin{example} The machine $M_1$ (Figure \ref{figure:example-avass}) is not strongly monotone since we have: $(q_1,0) \xrightarrow{x' = 19 - x} (q_1,19)$. However, we see that $Post^*(q_1, 10) =  \{(q_1, 9), (q_1, 10) \}$.
Therefore we can deduce that $(q_1, 10)$ cannot cover $(q_1,19)$. Hence $M_1$ is not well structured. We give, in Section \ref{section:avass}, an algorithm for deciding whether a $1$-AVASS is well structured.
\end{example}

It is shown in \cite{FINKEL200163} that \emph{almost} every transition system can be turned into a WSTS for the \emph{termination ordering} which is not, in general, decidable. So the problem is not only to decide whether a system is a WSTS in general; we have to choose a \emph{decidable} ordering. We show that deciding whether arbitrary (non-effective) transition systems are well-structured for the usual (decidable) ordering on natural numbers is undecidable.

\begin{proposition} \label{prop:wsp-unde-arith-counter-machine}
The well structured problem for $1$-arithmetic counter machines is undecidable.
\end{proposition}

\begin{proof}
	Since first order (FO) logic is undecidable, we can have a reduction from decidability of FO to checking whether a given arithmetic counter machine is well-structured. Let $\phi$ be a given FO formula with no free variables. Define the $1$-arithmetic counter machine $M=(\{q_0\},  \{\phi_0\}, \rightarrow)$, where $\phi_0 = ((x_1 = 0 \land y_1 = 2) \lor \phi)$ and $\rightarrow = \{ (q_0, \phi_0, q_0) \}$ and let $S_M$ be its associated transition system. Hence, if $\phi$ is a tautology, then for all $m, n \geq 0$, the transition $(q_0, m) \rightarrow (q_0,n)$ exists  in $S_M$. Hence $S_M$ is a well-structured transition system. However, if $\phi$ is false, then $S_M$ is not well-structured since there is a transition $(q_0, 0) \rightarrow (q_0, 2)$, but there is no transition from $(q_0, 1) \geq (q_0,0)$ which violates monotony.
	Hence $S_M$ is a WSTS iff $\phi$ is a tautology.
\end{proof}

We now show that restricting to \emph{effective} transition systems does not allow us to decide the property of being a WSTS.

\begin{corollary}
The well structured problem (for the usual ordering on $\mathbb{N}$) for effective transition systems whose set of configurations is included in $\mathbb{N}$ is undecidable.
\end{corollary}

\begin{proof} 
There exists a reduction from the Halting Problem as follows:

Given a Turing machine $M$, we define a transition system $S_M=(\mathbb{N}, \rightarrow_M)$ as follows: \\
If $(m = 0) \lor (M \text{ does not halt in } m \text{ steps})$, then, for all $n$, there is a transition $m \rightarrow_M n$. Hence this transition relation $\rightarrow_M$ is decidable.
Now, if $M$ does not halt, then there is a transition $m \rightarrow_M n$ for all $m,n \in \mathbb{N}$. This satisfies monotony, hence in this case, $S_M$ is a WSTS.
However, if $M$ halts in exactly $m$ steps, then there is no transition from $m+1$ but there is, in any case, a transition from $0$ to $n$ for all $n$.  Hence in this case, $S_M$ is not a WSTS.
Therefore, $S_M$ is a WSTS iff $T$ does not halt.
\end{proof}

\subsection{Testing whether an ordering is well}

In the previous results, the usual well ordering on natural numbers is not necessarily the unique \emph{decidable} ordering when considering the well structured problem for counter machines. Let $\leq$ be a decidable quasi ordering relation on $\mathbb{N}^d$. If we are interested in whether a counter machine with this ordering is WSTS, it raises the natural question of whether we can decide if $\leq$ is a wqo. Unfortunately, but unsurprisingly, we first show that this property is undecidable in dimension one ($d=1$).

\begin{proposition} \label{prop:gen-ordering-wqo-undec}
The property for a decidable ordering on $\mathbb{N}$ to be a well ordering is undecidable.
\end{proposition}

\begin{proof}
	We will have a reduction from Halting Problem to show undecidability of checking whether a relation on $\mathbb{N}$ is a wqo.
	
	Let $M$ be a Turing machine and $\leq_M$ its associated decidable relation defined as follows. For all $i,j$: we have $i \leq_M i$ and 
	$i \leq_M i+j$ iff $M$ does not halt in $i+j$ steps and $i+j \leq_M i$ if $M$ halts in at most $i+j$ steps; hence $\mathbb{N}$ is totally ordered by the decidable ordering $\leq_M$. If $M$ does not halt, we have $1 \leq_M 2 \leq_M ...  \leq_M i \leq_M i+1 \leq_M ...$ so $(\mathbb{N}, \leq_M)$ is a \textit{well ordering}. If $M$ halts in exactly $n$ steps, then there is an infinite strictly decreasing sequence $n >_M n+1 >_M n+2 >_M ...$, hence $(\mathbb{N},\leq_M)$ is not a well ordering because it is not well-founded. Therefore, checking whether an ordering $\leq$ encodes a well ordering is undecidable.
\end{proof}

Let us study the case of \emph{Presburger-definable orderings} in $\mathbb{N}^d$.
Among many equivalent characterizations of wqo, we know that a quasi ordering is well iff it satisfies well-foundedness and the finite anti-chain property. Both of these properties can be expressed using monadic second order variables. 
But, it is shown in \cite{DBLP:journals/corr/HorbachVW17} that Presburger Arithmetic with a single monadic variable becomes undecidable. Hence, this cannot directly be used to check if a Presburger-definable ordering is a \emph{wqo}. However, we should take a look at Ramsey quantifiers.

\begin{definition}
	$FO^{2-ram}$ is an extension of first order logic with 2-Ramsey quantifiers. A 2-Ramsey quantifier, denoted by $\exists^{2-ram} x y \phi(x,y)$ is satisfied if there exists an infinite subset of the domain, such that every pair of elements $(x,y)$ of this subset satisfies the formula $\phi(x,y)$.
\end{definition}

For the following section, we will let $\phi(x_1, ..., x_d, y_1, ..., y_d)$ denote a special formula for which we want to check if it encodes a well quasi order.

We will look at the structure $(\mathbb{N}^d, \phi, \leq)$ where $\leq$ denotes the usual ordering on $\mathbb{N}^d$, ie $(x_1, ..., x_d) \leq (y_1, ..., y_d)$ iff $x_i \leq y_i$, for all $i= 1,...,d$. Here $\phi$ is interpreted with two free variables in $\mathbb{N}^d$, ie $\phi( (x_1, ..., x_d), (y_1, ..., y_d))$ iff $\phi(x_1, ..., x_d, y_1, ..., y_d)$ holds.

Since $(\mathbb{N}^d, \phi, \leq)$ is an automatic structure and 2-Ramsey quantifiers preserve regularity for automatic structures (a detailed survey about automatic structures and Ramsey quantifiers can be found in \cite{rubin_2008}), we deduce the following Theorem.

\begin{theorem} \label{fo-ram-logic-dec}
	The structure $FO^{2-ram} (\mathbb{N}^d, \phi, \leq)$ is decidable.
\end{theorem}

Using theorem \ref{fo-ram-logic-dec} we can prove the following:

\begin{proposition} \label{pres-relation-wqo}
The property for a Presburger relation on $\mathbb{N}^d$ to be a well quasi ordering is decidable.
\end{proposition}

\begin{proof}
Given a Presburger formula $\phi(x_1, ..., x_d, y_1, ..., y_d)$, we can check if it encodes a quasi ordering since reflexivity and transitivity are Presburger-expressible.

To check finite anti-chain property, we define the formula $\psi_1$ in $FO^{2-ram} (\mathbb{N}^d, \phi, \leq)$ as follows:
$$\psi_1 = \exists^{2-ram} x y (\lnot \phi(x, y))$$

We observe that $\psi_1$ holds iff the ordering encoded by $\phi$ violates the finite anti-chain property.

To check the well-foundedness property, we define the formula $\psi_2$ in $FO^{2-ram} (\mathbb{N}^d, \phi, \leq)$ as follows:
$$\psi_2 = \exists^{2-ram} x y (\phi(y, x) \lor \lnot (x \geq y))$$

If $\psi_2$ is satisfied, then there exists an infinite set $T \subseteq \mathbb{N}^d$ such that any pair of elements $(x,y)$ with $x,y \in T$,  satsifies $\phi(y, x) \lor \lnot (x \geq y)$. Consider any enumeration of $T$, say $(t_1, t_2, ...)$.  Since $\leq$ is a wqo, we can derive a subsequence $(t_{i_1}, t_{i_2}, ...)$ such that $t_{i_n} < t_{i_{n+1}}$. Since each $(t_{i_n}, t_{i_{n+1}})$ satisfies $\phi(y, x) \lor \lnot (x \geq y)$, hence we can conclude that $\phi(t_{i_{n+1}}, t_{i_n})$ holds for all $n$. Hence the quasi ordering encoded by $\phi$ violates the well-foundedness property.

Similarly, if the quasi ordering encoded by $\phi$ violates the well-foundedness property, then there exists an infinite sequence $(t_1, t_2, ...)$ such that $\phi(t_{n+1}, t_n)$ for all $n \geq 1$. We can again derive a subsequence $(t_{i_1}, t_{i_2}, ...)$ such that $t_{i_n} \leq t_{i_{n+1}}$ and $t_{i_n} \neq t_{i_{n+1}}$ for all $n \geq 1$. 

Now the set $\{ t_{i_n} \mid n \in \mathbb{N} \}$ satisfies the formula $\psi_2$.

Hence we have that $\psi_2$ is satisfiable iff $\phi$ violates the well-foundedness property.

Since the structure $FO^{2-ram} (\mathbb{N}^d, \phi, \leq)$ is decidable, we can check whether $\psi_1$ and $\psi_2$ hold, to decide if $\phi$ encodes a wqo.
\end{proof}

%

\section{The well structured problem for PCM} \label{section:well-structured-problem}

In the sequel, whenever we talk about PCM being WSTS, we will consider the usual quasi ordering on $Q \times \mathbb{N}^d$ defined in subsection \ref{subsection:wsts}.
We introduce a general technique to prove undecidability of checking whether a counter machine of some class is a WSTS. Let $S_0$ be the class of machines we are interested in. We will show reduction from reachability in Minsky machines.

\begin{lemma}\label{gen-undecidability}
Suppose we have a procedure which takes a 2 counter Minsky machine with initial state $M = (Q, \rightarrow, q_0)$ and a control-state $q_1$ as input and generates a machine $N$ of class $S_0$ which satisfies the following two requirements:
\begin{itemize}
\newcounter{props}
\item All control-states in $M$ are reachable implies $N$ is a WSTS. \hfill  (\refstepcounter{props}\theprops\label{prop1})
\item $N$ is a WSTS implies $q_1$ is reachable in $M$ from $(q_0; 0, 0)$. \hfill (\refstepcounter{props}\theprops\label{prop2})
\end{itemize}

Then, the well structured problem for $S_0$ is undecidable.
\end{lemma}

\begin{proof}
Suppose that  the well structured problem for $S_0$ is decidable.
We will use the above procedure to get an algorithm for Minsky machine reachability.
Fix $(M, q_1)$, where $M = (Q, \rightarrow_M, q_0)$. We want to check if $q_1$ is reachable from $(q_0; 0,0)$.

Let $|Q| = n$. Consider all $2^{n-2}$ subsets $Q' \subseteq Q$ satisfying that $\{q_0, q_1\} \subseteq Q'$.
For each such $Q'$, let $\rightarrow_{Q'}$ denote the restriction of $\rightarrow_M$ to the set $Q' \times Q'$.
Hence, we can associate a Minsky machine $M' = (Q',\rightarrow_{Q'} ,q_0)$ to each such subset $Q'$. We call $M'$ a \emph{sub-machine} of $M$ corresponding to $Q'$.

Now, for each sub-machine $M'$, we consider the machine $N'$ of class $S_0$, generated by the given procedure from $(M', q_1)$. If there exists $M'$ such that $N'$ is a WSTS, then we have that $q_1$ is reachable in $M'$ (by condition (\ref{prop2})), hence in $M$.

On the other hand, if $q_1$ was reachable in $M$, then let $Q_{reach} \subseteq Q$ be the set of all control-states of $M$ which are reachable from $(q_0; 0,0)$. Let its corresponding sub-machine be $M'$. Since all control-states of $M'$ are reachable (by choice of $Q_{reach}$), therefore the corresponding $N'$ will be a WSTS (by condition (\ref{prop1})).

Hence, $q_1$ is reachable in $M$ from $(q_0; 0, 0)$ iff there exists a subset $Q' \subseteq Q$ satisfying that $\{q_0, q_1\} \subseteq Q'$ such that the corresponding sub-machine $M'$ is a WSTS. Since there are only $2^{n-2}$ such subsets, we can check all of them to decide whether $q_1$ is reachable in $M$.

Hence, we have given an algorithm to check reachability in Minsky machine. Therefore, the well structured problem for $S_0$ is undecidable.
\end{proof}

We will use Lemma \ref{gen-undecidability} to prove that the well structured problem for functional $1$-dim PCMs is undecidable. To apply Lemma \ref{gen-undecidability}, we need to give an algorithm which takes a Minsky machine $M = (Q, \rightarrow_M, q_0)$ and a control-state $q_1$, and generates a functional $1$-dim PCM $N_1$ satisfying conditions (\ref{prop1}) and (\ref{prop2}).

\paragraph*{Construction of a functional $1$-dim PCM $N_1$: } \label{construction-n1} Let $(M, q_0)$ be given. The procedure to generate a $1$-dim PCM $N_1$ is as follows:

Let $v_p(n)$ denote the largest power of $p$ dividing $n$. For $M = (Q, \rightarrow_M, q_0)$, we define the $1$-PCM $N_1 = (Q, \rightarrow_N, (q_0, 1))$ with the same set $Q$ of control-states. We will represent the values of the two counters $(m,n)$ by the one-counter values $2^m3^nc$ for any $c$ such that $v_2(c) = v_3(c) = 0$. Conversely, a configuration $(q, n)$ of $N_1$ will correspond to $(q; v_2(n), v_3(n))$ of $M$. Note that, we are allowing multiplication by constants $c$ in $N_1$ as long as $v_2(n)$ and $v_3(n)$ remain unchanged.

Increment/decrement of counters corresponds to multiplication/division by $2$ and $3$ which is Presburger expressible. Similarly, zero-test corresponds to checking divisibility by $2$ and $3$ which is again Presburger-expressible. So first, for each transition in $\rightarrow_M$, we add the corresponding transition to $\rightarrow_N$.

Now, to get the suitable properties of conditions (\ref{prop1}) and (\ref{prop2}), we will add two more types of transitions to $\rightarrow_N$. For each control-state $q$, we add a transition $(q, x_1' = 6x_1 + 1, q_0)$ to $\rightarrow_N$. We shall call it a \emph{"reset-transition"} because $v_2(6x_1+1) = v_3(6x_1+1) = 0$, so this transition corresponds to a counter-reset in $M$ from anywhere regardless of our present configuration. Note that such a transition would not change the reachability set in $M$. This "reset-transition" is crucial in forcing well-structuredness in $N$. Also, we add a transition $(q_0, (x_1 = 0 \land x_1' = 0), q_1)$ to $\rightarrow_N$ to ensure condition (\ref{prop2}). Since the configuration $(q_0, 0)$ cannot be reached from the initial configuration $(q_0, 1)$ during any run of $N_1$, this will also not affect the reachability set of $N_1$. Note that, all of our transitions are functional, hence $N_1$ is a functional $1$-dim PCM.

Now, we show that the construction of $N_1$ satisfies conditions (\ref{prop1}) and (\ref{prop2}).

\begin{lemma} \label{n1-prop-1-lemma}
The functional $1$-dim PCM \hyperref[construction-n1]{$N_1$} satisfies condition (\ref{prop1}).
\end{lemma}

\begin{proof}
	Suppose that all control-states of $M$ are reachable from $(q_0; 0, 0)$. Then we claim that $N_1$ will be a WSTS. Suppose there is a transition $(q, n) \rightarrow_{N} (q', m)$ and $(q, n')$ is a configuration with $(q, n') \geq (q, n)$. Hence we want to show existence of some path $(q, n') \xrightarrow{*}_{N} (q', m') \geq (q', m)$.
	
	\begin{casesp}
		\item The transition $(q, n) \rightarrow_{N} (q', m)$ is a "reset-transition". Hence $q' = q_0$ and $m = 6n+1$. In this case, note that since $n' \geq n$, the transition $(q, n') \rightarrow_{N} (q_0, 6n'+1) \geq (q_0, m)$ satisfies the requirement.
		
		\item The transition $(q, n) \rightarrow_{N} (q', m)$ is not a "reset-transition". In this case, $m \leq 3n$ because the above transition corresponds, in $M$ to an increment/decrement in $c_1$ or $c_2$ or a zero-test. In each case, we can check that $m \leq 3n$. Let there be a path $(q_0; 0, 0) \xrightarrow{*}_M (q'; n_1, n_2)$ in $M$ for some $n_1$, $n_2$. Such a path exists because all control-states in $M$ are reachable. Hence, we take the "reset-transition" $(q, n') \rightarrow_N (q_0, 6n'+1)$ and follow the corresponding path $(q_0, 6n'+1) \xrightarrow{*}_N (q', 2^{n_1}3^{n_2}(6n'+1)) \geq (q', 3n) \geq (q', m)$. Hence we have again shown monotony to prove that $N_1$ is a WSTS.
	\end{casesp}
	
	Hence we have shown that if all control-states of $M$ are reachable, then $N_1$ is monotone.
\end{proof}

\begin{lemma} \label{n1-prop-2-lemma}
The functional $1$-dim PCM \hyperref[construction-n1]{$N_1$} satisfies condition (\ref{prop2}).
\end{lemma}

\begin{proof}
	Since there is a transition $(q_0, 0) \rightarrow_{N} (q_1, 0)$, we deduce that if $N_1$ is a WSTS, then $(q_0, 1) \xrightarrow{*}_N (q_1, n)$ for some $n$ by monotony because $(q_0, 0) \leq (q_0, 1)$. Also note that since $N_1$ simulates $M$, hence reachability of $q_1$ in $N_1$ implies that $q_1$ is reachable from $(q_0; 0, 0)$ in $M$.
\end{proof}

Since we have provided a construction of functional $1$-dim PCM $N_1$ satisfying conditions (\ref{prop1}) and (\ref{prop2}), from Lemma \ref{gen-undecidability} we have that:

\begin{theorem}\label{1-dim-pcs-wsts}
The well structured problem for functional 1-dim PCMs is undecidable.
\end{theorem}

Similarly, we can use Lemma \ref{gen-undecidability} to show this result for 2 counter Minsky machines. This construction is as follows:

\subsection*{Well structured problem for 2-Minsky machines} \label{wsp:minky-machine:appendix}

We shall first show a construction for 4 counter Minsky machines satisfying conditions (\ref{prop1}) and (\ref{prop2}).

\begin{figure}
	\begin{subfigure}[b]{\textwidth}
		\centering
		\resizebox{0.9\linewidth}{!}{
			\begin{tikzpicture}[scale=0.2, every node/.style={scale=0.7}]
			\tikzstyle{every node}+=[inner sep=0pt]
			\node at (4, 15) [font=\huge] {$M{:}$};
			
			\node(t0) at (10,15) [circle, draw, minimum size=1.2cm] {$q$};
			\node(t1) at (20,15) [circle, draw, minimum size=1.2cm] {$q'$};
			\node(t2) at (20,10) [circle, draw, minimum size=1.2cm] {$q''$};
			
			\draw[->] (t0) -- node [above=0.1cm] {$c_1{=}0?$} (t1);
			\draw[->] (t0) -- node [above=0.05cm, sloped] {$c_1{\neq}0?$} (t2);
			
			\node at (15, 9) [rotate=270, font=\Large] {$\implies$};
			
			\node at (-16, 0) [font=\huge] {$N{:}$};
			\node(s0) at (-10,0) [circle, draw, minimum size=1.5cm] {$q$};
			\node(s1) at (0,0) [circle,draw, minimum size=1.5cm] {};
			\node(s2) at (10,0) [circle,draw, minimum size=1.5cm] {};
			\node(s3) at (20,0) [circle,draw, minimum size=1.5cm] {};
			\node(s4) at (35,0) [circle,draw, minimum size=1.5cm] {};
			\node(s5) at (35,-10) [circle,draw, minimum size=1.5cm] {};
			\node(s6) at (45,0) [circle,draw, minimum size=1.5cm] {$q'$};
			\node(s7) at (45,-10) [circle,draw, minimum size=1.5cm] {$q''$};
			
			\draw[->] (s0) -- (s1);
			\draw [->] (s1) to [out=105, in=75, looseness=6] node [midway, above=0.1cm] {$c_4\dec$} (s1);
			\draw[->] (s1) -- node [above=0.1cm] {$c_4{=}0?$} (s2);
			\draw [->] (s2) to [out=105, in=75, looseness=6] node [midway, above=0.1cm] {$c_1\dec,c_3\dec,c_4\inc$} (s2);
			\draw[->] (s2) -- node [above=0.1cm] {$c_3{=}0?$} (s3);
			\draw[->] (s3) -- node(mid) [above=0.1cm] {$c_1{=}0?$} (s4);
			\draw[->] (s3) -- node [near start, above=0.1cm, sloped] {$c_1{\neq}0?$} (s5);
			\draw [->] (s4) to [out=105, in=75, looseness=6] node [midway, above=0.1cm] {$c_1\inc,c_3\inc, c_4\dec$} (s4);
			\draw [->] (s5) to [out=105, in=75, looseness=6] node [midway, above] {$c_1\inc,c_3\inc, c_4\dec$} (s5);
			\draw[->] (s4) -- node [above=0.1cm] {$c_4{=}0?$} (s6);
			\draw[->] (s6) to [out=210, in=330] node [below=0.1cm] {$c_4{=}0?$} (s4);
			\draw[->] (s5) -- node [above=0.1cm] {$c_4{=}0?$} (s7);
			\draw[->] (s7) to [out=210, in=330] node [below=0.1cm] {$c_4{=}0?$} (s5);
			\end{tikzpicture}
		}
		\caption{Showing equivalent circuits for zero-tests in $M$}
		\label{c_1=0-figure}
	\end{subfigure}
	\begin{subfigure}[b]{\textwidth}
		\centering
		\resizebox{!}{!}{
			\begin{tikzpicture}[scale=0.2, every node/.style={scale=0.7}]
			\tikzstyle{every node}+=[inner sep=0pt]
			\node(s1) at (35,0) [circle,draw, minimum size=1.5cm] {};
			\node(s2) at (52,0) [circle,draw, minimum size=1.5cm] {};
			\node(s3) at (63,0) [circle,draw, minimum size=1.5cm] {$q_0$};
			
			\draw [->] (s2) to [out=105, in=75, looseness=6] node [midway, above] {$c_1\inc,c_2\inc,c_3\inc$} (s2);
			
			\draw [->] (s1) to [out=60, in=30, looseness=6] node [midway, above=0.1cm] {$c_1\dec$} (s1);
			\draw [->] (s1) to [out=120, in=150, looseness=6] node [midway, above=0.1cm] {$c_2\dec$} (s1);
			\draw [->] (s1) to [out=210, in=240, looseness=6] node [midway, below=0.1cm] {$c_3\dec$} (s1);
			\draw [->] (s1) to [out=330, in=300, looseness=6] node [midway, below=0.1cm] {$c_4\dec$} (s1);
			
			\draw [->] (s1) -- node [above=0.1cm] {$c_1{=}c_2{=}c_3{=}c_4{=}0?$} node [below=0.15cm] {$c_1\inc,c_2\inc,c_3\inc$} (s2);
			
			\draw [->] (s2) -- (s3);
			\draw [->] (29,0) -- (s1);
			\end{tikzpicture}
		}
		\caption{"Reset-circuit" for $N_2$}
		\label{reset-circuit-figure}
	\end{subfigure}
	\caption{Construction of a 4 counter Minsky machine $N_2$} 
	\label{4ctr-minsky-figure}
\end{figure}

\paragraph*{Construction of a 4 counter Minsky machine $N_2$: } \label{construction-n2}
Let $(M, q_1)$ be given. 
We will use $4$ counters to simulate $M$ in such a way that we can get all the desirable properties. We will use a configuration $(q;c_1, c_2, c_3, c_4)$ of $N_2$ to correspond to the configuration $(q; c_1-c_3, c_2-c_3)$ of $M$.

The procedure to generate a $4$-counter Minsky machine $N_2$ is as follows:

Let $M = (Q, \rightarrow_M, q_0)$. We define $N_2 = (Q_0, \rightarrow_N, (q_0; 1,1,1,0))$ where $Q_0$ is a superset of $Q$ as will be made clear. To get $N_2$, we will make the following modifications to $M$:

\begin{itemize}
	\item Replace zero-tests in $M$ with a circuit which checks if the respective counter equals $c_3$. We will use $c_4$ and add required additional control-states to implement such a test as illustrated in figure \ref{c_1=0-figure}.
	\item Now, from each control-state $q$, including the new ones added in previous step, add another circuit as illustrated in figure \ref{reset-circuit-figure} which allows one to reach $(q_0; n, n, n, 0)$ for any $n \geq 1$. Note that $q_0$ is the initial state of $M$. We shall again call it a "reset-circuit" since $(q_0; n, n, n, 0)$ in $N_2$ corresponds to $(q_0, 0, 0)$ in $M$. Hence this transition acts as a counter-reset. Note that adding such a transition to $M$ will not affect its reachability set. The "reset-circuit" will be used to ensure condition (\ref{prop1}).
	\item Finally, add a transition $(q_0; 0,0,0,0) \rightarrow_{N} (q_1; 0,0,0,0)$. This is to ensure property (\ref{prop2}). Note that we can never reach $(q_0; 0,0,0,0)$ from the initial configuration $(q_0; 1,1,1,0)$ in any run of $N_2$ since the "reset-circuit" resets to $(q_0; n, n, n, 0)$ for $n \geq 1$. Hence adding this transition does not affect the reachability set of $N_2$.
\end{itemize}

Now, we will show that the above construction indeed satisfies conditions (\ref{prop1}) and (\ref{prop2}).

\begin{lemma}
	The 4 counter Minsky machine \hyperref[construction-n2]{$N_2$} satisfies condition (\ref{prop1}).
\end{lemma}

\begin{proof}
	Suppose all control-states are reachable in $M$. Let there be a transition $(q; n_1, n_2, n_3, n_4) \rightarrow_N (q'; m_1, m_2, m_3, m_4)$ and a configuration $(q; n_1', n_2', n_3', n_4') \geq (q; n_1, n_2, n_3, n_4)$. Since all control-states in $M$ are reachable, there exists a path $(q_0; 0, 0) \xrightarrow{*}_M (q'; n_1, n_2)$. Choose $N \geq \max\{m_1, m_2, m_3, m_4\}$, and take the "reset-circuit" $(q; n_1', n_2', n_3', n_4') \xrightarrow{*}_N (q_0; N, N, N, 0)$. Then follow the corresponding path to $(q'; N+n_1, N+n_2, N, 0) \geq (q'; m_1, m_2, m_3, m_4)$ to satisfy monotony. Hence, if all control-states in $M$ are reachable it implies that $N_2$ is a WSTS.
\end{proof}

\begin{lemma}
	The 4 counter Minsky machine \hyperref[construction-n2]{$N_2$} satisfies condition (\ref{prop2}).
\end{lemma}

\begin{proof}
	Since $N_2$ is a WSTS, $(q_0; 0,0,0,0) \rightarrow_N (q_1; 0,0,0,0) \implies (q_0; 1,1,1,0) \xrightarrow{*} (q_1; n_1, n_2, n_3, n_4)$ which implies $q_1$ is reachable in $M$ since $N_2$ simulates $M$.
\end{proof}

\begin{theorem} \label{2-ctr-minsky-machine}
	The well structured problem for $2$-dim Minsky machines is undecidable.
\end{theorem}

\begin{proof}
	Since we have given the appropriate construction, by Lemma \ref{gen-undecidability}, we have that checking whether 4-counter Minsky machines are WSTS is undecidable.
	
	We will use the fact that a $d$-dim Minsky machine can be simulated by a $2$-counter Minsky machine \cite{Minsky:1967:CFI:1095587}.
	Let $N_2$ be the corresponding 4-counter Minsky machine for $(M, q_1)$ which satisfies conditions (\ref{prop1}) and (\ref{prop2}).
	Let $N_2'$ be the 2-counter Minsky machine which simulates $N_2$. Let $N_2''$ be the 2-counter machine obtained by adding another "reset" circuit to $N_2'$ which allows reachability to $(q_0; 0,0)$ from any configuration. Then, $N_2''$ will satisfy conditions (\ref{prop1}) and (\ref{prop2}) because $N_2''$ is simulating $N_2$ which satisfies conditions (\ref{prop1}), (\ref{prop2}) and we are allowing reset to initial configuration in $N_2''$. Hence, by Lemma \ref{gen-undecidability} again we have that the well structured problem for $2$-dim Minsky machines is undecidable.
\end{proof}

Now, we make the observation that we can perform zero-tests using affine functions. 
%
The basic idea is that a transition $x' = -x$ is only satisfied by a counter whose value is $0$. 
Increments/decrements can already be implemented in $2$-AVASS since translations are affine functions. A zero test on the first counter can be done by having a transition labelled by $\left(
  \left[ {\begin{array}{cc}
   -1 & 0 \\
   0 & 1 \\
  \end{array} } \right],  \left[ {\begin{array}{cc} 0 \\ 0
  
  \end{array} }\right]
  \right)$, and similarly for second counter. Since we can implement both increment/decrements and zero-tests with $2$-AVASS, we can simulate $2$-counter Minsky machines with $2$-AVASS. Note that we can extend this result to $d$-AVASS simulating $d$-counter Minsky machines.

As a direct consequence of this and Theorem \ref{2-ctr-minsky-machine}, we have that:

\begin{corollary} \label{aff-2-vass-undec}
The well structured problem for $2$-AVASS is undecidable.
\end{corollary}

However, if we consider strong monotony instead of monotony, the above undecidability results can be turned into a decidability result. Strong monotony can be expressed in Presburger arithmetic as follows:

\begin{align*}
\bigwedge\limits_{\phi \in \Phi} ( &\forall x_1 ... \forall x_d \forall x'_1 ... \forall x'_d \forall y_1 ... \forall y_d( (\bigwedge_{i=1}^d x_i \leq y_i) \land \phi(x_1, ..., x_d, x'_1, ..., x'_d) \\
& \implies (\exists y_1' ... \exists y_d' (\bigwedge_{i=1}^d x'_i \leq y_i') \land \phi(y_1, ..., y_d, y_1', ..., y_d'))) )
\end{align*}

Since Presburger arithmetic is decidable, the strong well structured problem for $d$-PCM is decidable.

\begin{theorem}
The strong well structured problem for $d$-PCM is decidable.
\end{theorem}

\begin{remark}
The validity of the formula of strong monotony can also be decided for extended PCM defined in decidable extensions of Presburger Arithmetic.
\end{remark}
%

\section{Decidability results for $1$-AVASS} \label{section:avass}

%
%

Now, let us look at some reachability and coverability results for the various models of AVASS.
First, we can simulate $2$-counter Minsky machines with $2$-AVASS. Since coverability and reachability are undecidable for $2$-counter Minsky machines, we directly have the following result:

\begin{corollary} \label{2-avass-cov-undec}
Control-state reachability, hence coverability is undecidable for $2$-AVASS.
\end{corollary}

Similarly, we showed in Construction of functional $1$-PCM \hyperref[construction-n1]{$N_1$} that we can also simulate $2$-counter Minsky machines with functional $1$-PCM. Hence, we also have the following:

\begin{corollary} \label{1-pcm-cov-undec}
Control-state reachability, hence coverability is undecidable for functional $1$-PCM.
\end{corollary}

Now, let us examine the case of $1$-AVASS. For $1$-AVASS, reachability and consequently coverability is decidable from work done in \cite{FGH-mfcs13}. We show that checking whether it is a WSTS is also decidable. Moreover, we give a simpler proof of decidability of reachability and coverability.

Given $M = (Q, \rightarrow)$ a $1$-AVASS and a final configuration $(q_f, n_f)$ that we want to check reachability for, we present Algorithm \ref{1-aff-vass-reachability-algo} which computes $Pre^*(q_f,n_f)$ as a Presburger formula. A transition $(q, x' = ax+b, q')$ is \textit{positive} if $a \geq 0$. Let a cycle/path in $M$ be called \textit{positive} if all transitions are positive. A cycle  $(q_1, ..., q_k, q_1)$ is called a \emph{simple cycle} if $q_1, ..., q_k$ are all pairwise distinct.

Let us denote by $Pre_q$ the set $Pre^*(q_f,n_f) \cap (\{q\} \times \mathbb{N})$. For a transition $t = (q, x' = ax+b, q')$ and a given subset of $X \subseteq \mathbb{N}$, let $Pre^t(X)$ denote $\{n : an+b \in X\}$. For a simple cycle $c$ rooted at $q$ with an effective guard and transition, extend the above notation $Pre^{c^i}(X)$ for $i$ repetitions of the cycle. Then, let $Pre^{c^*}(X) := \cup_{i \in \mathbb{N}} Pre^{c^i}(X)$. We will conveniently replace $X$ by a formula which denotes a subset of $\mathbb{N}$.

\begin{algorithm}
\caption{Algorithm for computing $Pre^*(q_f,n_f)$ in $1$-AVASS} \label{1-aff-vass-reachability-algo}
\begin{algorithmic}[1]
\Procedure{computePre*}{}
\ForAll {$q \in Q$}
	\State $\phi_q \equiv \bot$
\EndFor
\State $\phi_{q_f} \equiv (n = n_f)$

\ForAll {$q \in Q$}
	\ForAll {simple cycles $c$ rooted at $q$}
		\State $c$.transition = \Call{simplifyTransition}{$c$} 
		\State $c$.guard = \Call{computeGuard}{$c$}
	\EndFor
\EndFor

\State notFinished = True
\While {notFinished}
	\State notFinished = False
	\ForAll {$q \in Q$}
		\State $\phi' = \phi_q$
		\ForAll {transitions $t = (q, x'=ax+b, q') \in \rightarrow$}
			\State \Call{ExploreTransition}{$t$}
			
		\EndFor
		\ForAll {simple cycles $c$ containing $q$}
			\State \Call{ExploreCycle}{$c$}
		\EndFor
		\If {$\phi' \neq \phi_q$}
		\Comment{Check equality as Presburger formulae}
			\State notFinished = True
		\EndIf
	\EndFor
\EndWhile
\EndProcedure
\end{algorithmic}
\end{algorithm}

The algorithm will keep a variable $\phi_q$ for each control-state $q \in Q$ which will store a Presburger formula (with one free variable $n$) denoting the currently discovered subset of $Pre_q$. Let this be denoted by $\llbracket \phi_q \rrbracket$, i.e. $\llbracket \phi_q \rrbracket := \{n : \phi_q(n)\}$. For uniformity, we can assume that $\phi_q$ is a disjunction of formulae of form $range \land mod$ where $range \equiv (r \leq n \leq s)$ ($s$ possibly $\infty$) and $mod \equiv (n =_{d_q} d)$.

We initially simplify each simple cycle into a meta-transition which is the composition of all individual transitions in the cycle. We will also compute the guard of a cycle. Since each positive transition has an upward closed guard and each negative transition has a downward closed guard, the guard of a cycle will be of the form $r \leq n \leq s$ for some $r, s \in \mathbb{N}$ ($s$ possibly $\infty$). Hence, we will only consider a cycle in terms of its guard and its meta-transition.

We use two main procedures in \hyperref[1-aff-vass-reachability-algo]{\textproc{computePre*}}:
\begin{enumerate}
\item \textproc{ExploreTransition}: Given a transition $t = (q, x'=ax+b, q')$, it computes $Pre^t(\phi_{q'})$ and appends it to $\phi_q$.
\item \textproc{ExploreCycle}: Given a simple cycle $c$ rooted at $q$, it computes $Pre^{c^*}(\phi_q)$ and appends it to $\phi_q$.
\end{enumerate}

\begin{lemma} \label{1-avass-pre-c*-computable}
For any transition $t$, and any simple cycle $c$, given $\phi_q$, $Pre^t(\phi_q)$ and $Pre^{c^*}(\phi_q)$ are both Presburger expressible and effectively computable.
\end{lemma}

\begin{proof}
	Let $\phi \equiv \lor_{i=1}^k (range_i \land mod_i)$. First we note that $Pre^{c^*}(\phi_q) = \cup_{i=1}^k Pre^{c^*}(range_i \land mod_i)$, so we will focus on a single $range \land mod$ clause in $\phi$.
	
	For a transition $t$, we can compute $Pre^t$ by looking at each $range \land mod$ clause in $\phi_q$, and computing its inverse. 
	
	Similarly, we can compute $Pre^{c^i}$ for any $i \in \mathbb{N}$. We will show that $Pre^{c^*}(range \land mod)$ is also computable.
	
	Given a cycle $c$ rooted at $q$, with guard $r \leq n \leq s$ ($s$ possibly $\infty$) and the meta-transition $y = ax+b$, let $range \equiv (r_1 \leq n \leq s_1)$ and $mod \equiv (n =_d d_1)$. By $Post^{c^i}(n)$ we denote $i$ successive applications of a cycle to $n$. To show that $Pre^{c^*}(\phi_q)$ is effectively computable, we will look at multiple cases:
	
	\begin{casesp}
		\item $s < \infty$. \\
		In this case, since the cycle can only be fired by a finite number of inputs, we can simply take each input $n$ and compute $Post^{c^i}(n)$  for all $i$ till $c$ can no longer be activated, or it repeats. Then, we can decide whether or not $n \in Pre^{c^*}(range \land mod)$ based on if any of the reachable values satisfy $range \land mod$.
		
		\item $s = \infty$, $a \leq 1$. \\
		Note that since the guard is upward closed, it implies that the cycle is positive. Hence the net transition is $x' = ax+b$ for $a \geq 0$. If $a = 0$, we are done. If $a = 1$, then the cycle is a translation. In this case, $Pre^{c^*}(range \land mod)$ is again computable using modulo relations.
		
		\item $s = \infty$, $a \geq 2$. \\
		Let $N = \ceil{\frac{-b}{a-1}}$. The first thing to observe is that for all $n > N$, $Post^c(n) > n$. Hence, upon repeated application, $n$ keeps increasing. If $s_1 < \infty$, then $ n \geq \max\{s_1, N\} \implies n \notin Pre^{c^*}(range \land mod)$. With only finitely many values left to consider, we can again compute $Post^{c^i}(n)$ for all $r \leq n \leq \max\{s_1, N\}$ to determine whether $n \in Pre^{c^*}(range \land mod)$. If $s_1 = \infty$, then for values $n \geq N$, we only need to be concerned with their value $\mod d$. Hence, we can first compute the set of values $\ell_1, ..., \ell_k \mod d$ such that for some $j$, we have $Post^{c^j}(\ell_i) =_d d_1$. Then we know that $(n \geq N \land n =_d \ell_i) \implies n \in Pre^{c^*}(range \land mod)$. For $n \leq N$, we can again compute $Post^{c^i}(n)$ to determine whether $n \in Pre^{c^*}(range \land mod)$. Since it does not go off to infinity, it will again terminate or repeat.
	\end{casesp}
	
	Thus, we have shown that for any cycle $c$, we can compute $Pre^{c^*}(\phi_q)$.
\end{proof}

With this lemma the algorithm is well-defined. Now let us prove the termination and the correctness of the algorithm.

\begin{proposition}
Algorithm \hyperref[1-aff-vass-reachability-algo]{\textproc{computePre*}} terminates.
\end{proposition}

\begin{proof}
For each $q$, we will show that $Pre_q$ can be obtained in finitely many iterations of the algorithm. Let $q \in Q$ be arbitrary.

\begin{casesp}
\item $Pre_q$ is finite: \\
Each value will be discovered in finitely many iterations, hence $Pre_q$ will be obtained in finitely many iterations.

\item $Pre_q$ is infinite: \\
Since we are talking about reaching $(q_f, n_f)$, we note that the only transitions which can decrease arbitrarily large values are transitions of the form $x' = b$ or $x' = x - a, a > 0$. Hence, since $Pre_q$ has arbitrarily large values, and each run has to reach $n_f$ (i.e. has to be decreased), we can see that there must either be a transition $x'=b$, or a positive cycle with meta-transition $x' = x-a$, reachable from $q$ through a positive path.
	\begin{casesp}
	\item There is a transition $x'=b$: \\
	In this case, there exists $N$ such that for all $n \geq N$, the same path suffices. In this case, once the aforementioned path is discovered, $\{ n : n \geq N \}$ becomes a subset of $\llbracket \phi_q \rrbracket \subseteq Pre_q$, which leaves finitely many values in $Pre_q \setminus \llbracket \phi_q \rrbracket$, which can again be discovered by finitely many additional runs.
	
	\item There are positive cycles with meta-transition $x' = x-a$: \\
	The idea is that we will cover $Pre_q$ when we compute $Pre^{c^*}$ for such a cycle $c$. This is because for such a cycle, all that matters is the value of the counter modulo $a$. Since there are only finitely many distinct values modulo $a$, these will again be discovered in finitely many runs. Hence, each cycle will be discovered in finitely many runs. Therefore since there are finitely many simple cycles, the corresponding values of $Pre_q$ will also be discovered in finitely many runs.
	\end{casesp}
\end{casesp}

Hence, for all $q$, in finitely many runs we will get $Pre_q = \llbracket \phi_q \rrbracket$. At such a point, the algorithm has to stop, hence termination is guaranteed.
\end{proof}

\begin{theorem} \label{aff-1-vass-pre*-computable}
(Correctness) Given a $1$-AVASS $M=(Q, \rightarrow)$ and a configuration $(q,n)$, the algorithm \hyperref[1-aff-vass-reachability-algo]{\textproc{computePre*}} computes $Pre^*(q,n)$ as a Presburger formula.
\end{theorem}

\begin{proof}
We will show that Algorithm \ref{1-aff-vass-reachability-algo} upon termination will always have $\llbracket \phi_q \rrbracket = Pre_q$.

That $\llbracket \phi_q \rrbracket \subseteq Pre_q$ should be clear. Suppose the algorithm terminates with $\llbracket \phi_q \rrbracket \subsetneq Pre_q$ for some $q \in Q$. For some value $n \in Pre_q \setminus \llbracket \phi_q \rrbracket$, consider a path which covers $(q_2, n_2)$, say the path is $(q,n) \rightarrow (p_1, n_1) \rightarrow ... \rightarrow (p_m, n_m) \rightarrow (q_2, n')$. In such a path, consider the largest $i$, such that $n_i \notin \llbracket \phi_{p_i} \rrbracket$. Now, in the last iteration of the algorithm, since $n_{i+1} \in \llbracket \phi_{p_{i+1}} \rrbracket$ (by choice of $i$), hence, we will explore the edge to include $n_i \in \llbracket \phi_{p_i} \rrbracket$. Hence, the algorithm would not have terminated. Contradiction. Hence, when the algorithm terminates, $\llbracket \phi_q \rrbracket = Pre_q$.
\end{proof}


\begin{example} 
Let us consider machine $M_1$ in Figure \ref{figure:example-avass}. Suppose we want to compute $Pre^*(q_1, 19)$. We begin with $\phi_{q_1} \equiv (n=19)$, $\phi_{q_2} \equiv \bot$. If we apply \textproc{ExploreTransition} to the transition $(q_2, (x'=x), q_1)$, we will get $\phi_{q_2} \equiv (n=19)$. If we now apply \textproc{ExploreCycle} to the cycle $(q_2, x'=x-3, q_2)$, we will get $\phi_{q_2} \equiv (n \geq 19 \land n =_3 1)$. Continuing like this, we end up with $\phi_{q_1} \equiv (n \in \{0, 3, 6, 19\} \lor (n \geq 13 \land n =_3 1) \lor (n \geq 32 \land n =_3 2) \lor (n \geq 45 \land n=_3 0))$ and $\phi_{q_2} \equiv (n \geq 0 \land n =_3 0) \lor (n \geq 19 \land n=_3 1) \lor (n \geq 32 \land n=_3 2)$. This is $Pre^*(q_1, 19)$.
\end{example}

\begin{corollary} \label{1-avass-reachability}\cite{FGH-mfcs13}
Reachability (hence coverability and control-state reachability) for $1$-AVASS is decidable.
\end{corollary}

\begin{proof}
	Suppose we want to check reachability of $(q_2, n_2)$ from $(q_1, n_1)$.
	Once we have computed $Pre^*(q_2, n_2)$, we can check easily whether $(q_1, n_1) \in Pre^*(q_2, n_2)$ to solve reachability. Once we have reachability, we can show a reduction from coverability to reachability to show that coverability is also decidable for $1$-AVASS, as follows.
	Suppose we want to check coverability of $(q_2, n_2)$. We can add two new control-states $q_3$ and $q_4$ and add the transitions $(q_2, (x' = x - n_2), q_3)$ and $(q_3, (x' = 0), q_4)$. Now, $(q_4, 0)$ is reachable iff $(q_2, n_2)$ coverable.
\end{proof}

\begin{remark}
Algorithm \ref{1-aff-vass-reachability-algo} also works if we extend the model of $1$-AVASS with Presburger guards at each transition. Hence, reachability, coverability and the well-structured problem are all decidable for this model as well.
\end{remark}

It could be useful to determine whether an $1$-AVASS is a WSTS (with strict monotony) because if it is the case, it will allow to decide other problems like the boundedness problem that is not immediately a consequence of the computability of $Pre^*(\suparrow(q,n))$. Since we can compute $Pre^*(q,n)$, we can also compute $Pre^*(\suparrow(q,n))$ by the same technique as in Corollary \ref{1-avass-reachability}. This can be used to determine whether a given $1$-AVASS is a WSTS as follows.
\begin{theorem} \label{1-aff-WSTS-dec}
The well structured problem is decidable for $1$-AVASS.
\end{theorem}

\begin{proof}
First we show that $M$ is a WSTS, iff for all negative transitions $(q_1, (x' = ax+b), q_2)$, the set $\{q_1\} \times \mathbb{N}$ is a subset of $Pre^*(\suparrow(q_2, b))$.
For any negative transition $(q_1, (x' = ax+b), q_2)$, we have $(q_1, 0) \rightarrow (q_2, b)$. If $M$ is a WSTS, by monotony, for any $n \geq 0$, there exists a path $(q_1, n) \xrightarrow{*} (q_2, b') \geq (q_2, b)$ because $(q_1, n) \geq (q_1, 0)$. This implies that $\{q_1\} \times \mathbb{N}$ is a subset of $Pre^*(\suparrow(q_2, b))$.

In the other direction, let there be a transition $(q_1, n) \rightarrow (q_2, an+b)$ and $(q_1, n') \geq (q_1, n)$. If the transition is positive, i.e. $a \geq 0$, then we directly have the transition $(q_1, n') \rightarrow (q_2, an'+b) \geq (q_2, an+b)$. If the transition is negative, then we have that $(q_2, an+b) \leq (q_2, b)$. Since $(q_1, n') \in Pre^*(\suparrow(q_2,b))$ (by hypothesis, since it is a negative transition), hence we have that $(q_1, n') \xrightarrow{*} (q_2, b') \geq (q_2, b) \geq (q_2, an+b)$. Hence, $M$ is monotone. Therefore, $M$ is a WSTS iff for all negative transitions $(q_1, (x' = ax+b), q_2)$, the set $\{q_1\} \times \mathbb{N}$ is a subset of $Pre^*(\suparrow(q_2, b))$.

Now, since $Pre^*(\suparrow(q, n))$ is computable, we can check that for each negative transition $(q_1, (x' = ax+b), q_2)$, the set $\{q_1\} \times \mathbb{N}$ is a subset of $Pre^*(\suparrow(q_2, n))$ to determine whether $M$ is a WSTS or not.
\end{proof}

\begin{example} 
Let us consider machine $M_1$ in Figure \ref{figure:example-avass} and its negative transition $(q_1, x'=19-x, q_1)$. We observe that the set $Pre^*(\suparrow(q_1, 19)) = \{q_1, q_2\} \times \{n: n \geq 19\}$ does not contain $\{q_1\} \times \mathbb{N}$, hence machine $M_1$ is not a WSTS.
However, in this example (Figure \ref{figure:example-avass}), if we replace the transition $(q_1, (x'=x-13),q_2)$ by $(q_1, (x'=x+1),q_2)$, we will get a new machine $M_2$ which is still not a $1$-VASS, but it is a WSTS.
\end{example}

Let us focus our attention to positive AVASS now. We know that for positive $1$-AVASS reachability is decidable from Corollary \ref{1-avass-reachability}. We show that reachability is undecidable for positive $2$-AVASS by reduction from Post's Correspondence Problem (PCP) \cite{DBLP:journals/corr/Halava14}. Our result completes the view about decidability of reachability for VASS extensions in small dimensions. As a matter of fact, reachability is undecidable for VASS with two resets in dimension $3$ (to adapt the proof in \cite{dufourd98}), hence for positive $3$-AVASS but it is decidable for VASS with two resets in dimension $2$ \cite{FLS-fsttcs18}. If we replace resets by affine functions, reachability becomes undecidable in dimension two.

Reichert gives in \cite{DBLP:phd/hal/Reichert15} a reduction from the Post correspondence problem to reachability in a subclass of $2$-AVASS and we may remark that his proof is still valid for \emph{positive} $2$-AVASS. Blondin, Haase and Mazowiecki made some similar observations \cite{DBLP:conf/concur/BlondinHM18} for subclasses of $3-\mathbb{Z}$-AVASS, with positive matrices. Our proof is essentially the same as \cite{DBLP:phd/hal/Reichert15}.

\begin{theorem} \label{pos-2-avass-reach-undec}
Reachability is undecidable for positive $2$-AVASS.
\end{theorem}

\begin{proof} 
Suppose we are given an instance of PCP, i.e. we are given $a_1, ..., a_k, b_1, ..., b_k \in \{0,1\}^*$ for some $k \in \mathbb{N}$. We want to check if there exists some sequence of numbers $n_1, ..., n_{\ell} \in \{1, ..., k\}$ such that $a_{n_1}...a_{n_{\ell}} = b_{n_1}...b_{n_{\ell}}$ (concatenated as strings).

We will construct the positive $2$-AVASS as demonstrated in Figure \ref{2-avass-pcp}, where $|a_i|$ refers to the length of the string, and $(a_i)_2$ refers to the number encoded by the string $a_i$ if read in binary (most significant digit to the left). The idea is that we use the two counters to store the value of $(a_{n_1}...a_{n_{\ell}})_2$ and $(b_{n_1}...b_{n_{\ell}})_2$ for any $n_1, ..., n_{\ell}$. But we first increment each counter to keep track of leading zeroes. Now, the configuration $(q_2; 0,0)$ is reachable from $(q_0; 0,0)$ in the positive $2$-AVASS described in Figure \ref{2-avass-pcp} iff the given PCP has an affirmative answer. Hence, checking reachability in positive $d$-AVASS is undecidable for $d \geq 2$.
\end{proof}

\begin{figure}
    \centering
        \resizebox{0.7\linewidth}{!}{
\begin{tikzpicture}
\tikzstyle{every node}+=[inner sep=0pt]
\node(s0) at (-4,0) [circle,draw, minimum size=1cm] {$q_0$};
\node(s1) at (0,0) [circle,draw, minimum size=1cm] {$q_1$};
\node(s2) at (7,0) [circle,draw, minimum size=1cm] {$q_2$};

\draw [->] (s2) to [out=105, in=75, looseness=6] node [midway, above] {$\left( \vec{I}, \begin{bmatrix} -1 \\ -1 \end{bmatrix} \right)$} (s2);
\draw [->] (s1) -- node [above=0.1cm] {$(\vec{I}, \vec{0})$} (s2);
\draw [->] (s0) -- node [above=0.1cm] {$\left( \vec{0}, \begin{bmatrix} 1 \\ 1 \end{bmatrix} \right)$} (s1);

\draw [->] (s1) to [out=60, in=30, looseness=15] node [midway, above right=0.2cm and -1cm] {$\left( \begin{bmatrix} 2^{|a_1|} & 0 \\ 0 & 2^{|b_1|} \end{bmatrix}, \begin{bmatrix}(a_1)_2 \\ (b_1)_2 \end{bmatrix} \right)$} (s1);
\draw [->] (s1) to [out=120, in=150, looseness=15] node [midway, above left=0.2cm and -1cm] {$\left( \begin{bmatrix} 2^{|a_2|} & 0 \\ 0 & 2^{|b_2|} \end{bmatrix}, \begin{bmatrix}(a_2)_2 \\ (b_2)_2 \end{bmatrix} \right)$} (s1);
\draw [->] (s1) to [out=210, in=240, looseness=15] node [midway, below left=0.2cm and -1cm] {$\left( \begin{bmatrix} 2^{|a_3|} & 0 \\ 0 & 2^{|b_3|} \end{bmatrix}, \begin{bmatrix}(a_3)_2 \\ (b_3)_2 \end{bmatrix} \right)$} (s1);
\draw [->] (s1) to [out=330, in=300, looseness=15] node [midway, below right=0.2cm and -1cm] {$\left( \begin{bmatrix} 2^{|a_k|} & 0 \\ 0 & 2^{|b_k|} \end{bmatrix}, \begin{bmatrix}(a_k)_2 \\ (b_k)_2 \end{bmatrix} \right)$} (s1);

\node [below=0.5cm of s1, font=\huge] {$...$};
\draw [->] ($(s0)+(-1.2,0)$) -- (s0);
\end{tikzpicture}
        }
\caption{Construction for undecidability of reachability for positive $2$-AVASS by reduction from PCP.}
\label{2-avass-pcp}
\end{figure}

Also, we note that positive-AVASS are well-structured with strong monotony. Hence coverability is decidable \cite{FMP-wstsPN-icomp}.
If we look at totally-positive AVASS, we can see that coverability is already decidable by the same argument. However, reachability is also decidable.

\begin{theorem} \label{tot-pos-avass-reach-dec}
Reachability is decidable in totally-positive AVASS for any dimension.
\end{theorem}

\begin{proof}
Let $M = (Q, \rightarrow)$ be a totally-positive $d$-AVASS. Given $(q_0; n_1, ..., n_d)$, suppose we want to check reachability of $(q_f; m_1, ..., m_d)$. Let $N = \max\{m_1, ..., m_d\}$. Let $f_N: \mathbb{N} \rightarrow \{1, ..., N, \omega \}$ be the function which is identity on $\{1, ..., N\}$ and maps $\{N+1, ... \}$ to $\omega$. Extend this function to the set $\mathbb{N}^d$ component-wise. Since $M$ is totally-positive, we can restrict our search space from $Q \times \mathbb{N}^d$ to $Q \times \{0, ..., N, \omega \}^d$ by applying $f_N$ to each configuration and using the following arithmetic rules:
$0.\omega = 0$, and  for all $k \geq 1$, $k.\omega = \omega$ and $\omega + k = \omega$.

We claim that if $(q_f; m_1, ..., m_d)$ is reachable, then it is reachable in this restricted search-space.
This follows from the fact that given any element $(n_1, ..., n_d)$ of $\mathbb{N}^d$, and a totally positive transition $t = (A, b)$, we will have that $t(f_N(n_1, ..., n_d)) = f_N(t(n_1, ..., n_d))$ ($t$ acts on $f_N(n_1, ..., n_d)$ to give an element in $\{0, ..., N, \omega\}^d$). This is because a totally positive transition cannot decrease a value other than by multiplying it by $0$, hence any value greater than $N$ will continue to be greater than $N$. 
Also note that, by choice of $N$, $f_N(m_1, ..., m_d) = (m_1, ..., m_d)$.

Once we have this, we can make an induction on the length of the path to see that if $(q_f; m_1, ..., m_d)$ is reachable, it is reachable in the restricted search-space $Q \times \{0, ..., N\}^d$.

Since $Q \times \{0, ..., N, \omega \}^d$ is finite, this shows decidability of reachability.
\end{proof}

\section{Conclusion and perspective}

\begin{figure}
    \centering
        \resizebox{0.7\linewidth}{!}{
\begin{tikzpicture}[]

\pgfdeclarelayer{wsts}
\pgfdeclarelayer{pre}
\pgfdeclarelayer{reach}
\pgfdeclarelayer{cover}
\pgfdeclarelayer{undec}
\pgfdeclarelayer{front}
\pgfsetlayers{undec,cover,reach,pre,wsts,front}

\begin{pgfonlayer}{front}

\node [draw, fill=white] (tot-pos-d-avass) at (0,0) {Totally positive $d$-AVASS};
\node [draw, fill=white] (1-avass) at (5,0) {$1$-AVASS};
\node [draw, fill=white] (1-minsky) at (5, -1) {$1$-Minsky machines};
\node [draw, fill=white] (vass) at (1.8, 1.5) {VASS};
\node [draw, fill=white] (pos-avass) at (1.8, 3) {Positive $d$-AVASS};
\node [draw, fill=white] (2-minsky) at (4, 4.5) {$2$-Minsky machines};
\node [draw, fill=white] (d-avass) at (0, 5)  {$d$-AVASS ($d \geq 2$)};

\draw [->] (1-minsky) -- (1-avass);
\draw [->] (vass) to (pos-avass);
\draw [->] (tot-pos-d-avass) to [out=90, in=260] ([xshift=-15] pos-avass.south);
\draw [->] (1-avass) .. controls ($(pos-avass)+(3,2)$) and ($(2-minsky)+(-2,-2)$) .. ([xshift=20] d-avass.south); 
\draw [->] (pos-avass) to (d-avass);
\draw [->] (2-minsky.west) to (d-avass.east);
\end{pgfonlayer}

\begin{pgfonlayer}{wsts}
\node[draw=yellow, fill=white, opacity=0, fill opacity=0.5, inner xsep=1pt, inner ysep=13, fit={(vass) (pos-avass) (tot-pos-d-avass)}] {};
\node[draw=yellow, fill=yellow, fill opacity=0.3, inner xsep=1pt, inner ysep=13, fit={(vass) (pos-avass) (tot-pos-d-avass)}] (wsts-ellipse) {};
\node[below = -0.4 of wsts-ellipse] {WSTS};
\end{pgfonlayer}

\begin{pgfonlayer}{pre}
\node[draw=green, fill=white, opacity=0, fill opacity=1, inner xsep=10, inner ysep=15, fit={(1-avass) (tot-pos-d-avass) (1-minsky)}] (pre-ellipse) {};
\node[draw=green, fill=green, fill opacity=0.2, inner xsep=10, inner ysep=15, fit={(1-avass) (tot-pos-d-avass) (1-minsky)}] (pre-ellipse) {};
\node[below = -0.5 of pre-ellipse, sloped] {Pre* computable};
\end{pgfonlayer}

\begin{pgfonlayer}{reach}
\node[draw=red, fill=white, opacity=0, fill opacity=1, inner xsep=10, inner ysep=15, fit={(pre-ellipse) (vass)}] (reach-D-ellipse) {};
\node[draw=red, fill=red, fill opacity=0.2, inner xsep=10, inner ysep=15, fit={(pre-ellipse) (vass)}] (reach-D-ellipse) {};
\node[below = -0.5 of reach-D-ellipse, sloped] {Reachability decidable};
\end{pgfonlayer}

\begin{pgfonlayer}{cover}
\node[draw=blue, fill=white, opacity=0, fill opacity=1, inner xsep=10, inner ysep=15, fit={(reach-D-ellipse) (pos-avass)}] (cover-D-ellipse) {};
\node[draw=blue, fill=blue, fill opacity=0.2, inner xsep=10, inner ysep=15, fit={(reach-D-ellipse) (pos-avass)}] (cover-D-ellipse) {};
\node[below = -0.5 of cover-D-ellipse, sloped] {Coverability decidable};
\end{pgfonlayer}

\begin{pgfonlayer}{undec}
\node[draw=gray, fill=white, fill opacity=1, inner xsep=10, inner ysep=15, fit={(cover-D-ellipse) (d-avass) (2-minsky)}] (undec-ellipse) {};
\node[draw=gray, fill=gray, fill opacity=0.2, inner xsep=10, inner ysep=15, fit={(cover-D-ellipse) (d-avass) (2-minsky)}] (undec-ellipse) {};
\node[below = -0.5 of undec-ellipse, sloped] {Coverability undecidable};
\end{pgfonlayer}
\end{tikzpicture}
        }
        \caption{Showing reachability and coverability results for various AVASS models.}
\label{avass-figure}
\end{figure}

We introduced two variants of the well structured problem for PCM and we solve it for many classes of PCMs. Moreover, we answer the decidability questions for reachability and coverability for classes of PCMs and AVASSs (we summarise the results of Section \ref{section:avass} in Figure \ref{avass-figure}).

Many open problems can be attacked like the complexity of reachability for $1$-AVASS (reachability is NP for $1$-VASS and PSPACE for polynomial VASS), the size of $Pre^*$ of a $1$-AVASS (and its relation with the theory of flattable VASS \cite{DBLP:conf/concur/LerouxS04}), and the decidability of the property for a Presburger relation on $\mathbb{N}^d$ to be a well-quasi ordering for $d \geq 2$.

We also open the way to study the decidability of the well structured problems (for various orderings) for many other models like pushdown counter machines, FIFO automata, Petri nets extensions. For instance, we wish to solve the  \emph{well structured problems} for FIFO automata. We know that \emph{lossy} FIFO automata are well structured (for the subword ordering) but what is the class of \emph{perfect} FIFO automata which is well structured (for the prefix ordering)?\\

\noindent\textbf{Acknowledgements.} We would like to thank Georg Zetzsche for
showing us how to extend our previous result about the decidability for a relation to be a well quasi ordering when the relation is defined by a Presburger formula in dimension one (Proposition 10 in \cite{GF-fsttcs19}) to arbitrary dimensions (Proposition \ref{pres-relation-wqo} in this paper).

\bibliography{biblio}

\end{document}